\documentclass[10pt]{article}
\usepackage[dvipsnames]{xcolor}
\usepackage[a4paper, margin=1in]{geometry}
\usepackage[utf8]{inputenc}
\usepackage{amsmath,amssymb,amsfonts,amsthm}
\usepackage{mathtools}
\usepackage{xspace}
\usepackage{enumerate}
\usepackage{cleveref}
\usepackage{thm-restate}

\theoremstyle{plain}
  \newtheorem{theorem}{Theorem}

  \newtheorem{lemma}[theorem]{Lemma}
  \newtheorem{proposition}[theorem]{Proposition}

  \newtheorem*{theorem*}{Theorem}
  \newtheorem*{corollary*}{Corollary}
  \newtheorem*{lemma*}{Lemma}
  \newtheorem*{proposition*}{Proposition}
  \newtheorem*{claim*}{Claim}

\theoremstyle{definition}

  \newtheorem*{definition*}{Definition}
  \newtheorem*{example*}{Example}
  \newtheorem*{question*}{Question}
  \newtheorem*{conjecture*}{Conjecture}

\usepackage{tikz}
\usetikzlibrary{arrows.meta}
\usetikzlibrary{calc}
\usetikzlibrary{bending}
\usetikzlibrary{shapes}

\usepackage{characters}
\newcommand*{\Int}{\mathbb{Z}}
\newcommand*{\Rat}{\mathbb{Q}}

\usepackage{mathtools}

\newcommand*{\csp}[1]{\textsc{CSP}\ensuremath{(#1)}}
\newcommand*{\mincsp}[1]{\textsc{MinCSP}\ensuremath{(#1)}}
\newcommand*{\wmincsp}[1]{\textsc{Weighted MinCSP}\ensuremath{(#1)}}
\newcommand*{\cwmincsp}[1]{\textsc{Compressed Weighted MinCSP}\ensuremath{(#1)}}
\newcommand*{\PA}{\ensuremath{<,\leq,=,\neq}}

\usepackage{boxedminipage}
\newcommand{\probDef}[4]{%
  \noindent
  \begin{center}
  \begin{boxedminipage}{\textwidth}
  {\sc #1}\\[5pt]
  \begin{tabular}{l p{0.75 \textwidth}}
  {\sc Instance}: & #2\\
  {\sc Parameter}: & #3\\
  {\sc Question}: & #4
  \end{tabular}
  \end{boxedminipage}
  \end{center}
}

\title{Parameterized Complexity of MinCSP over the Point Algebra} %

\author{
  George Osipov\thanks{Department of Computer and Information Science, Linköping University, Sweden, \texttt{george.osipov@pm.me}.} 
  \and
  Marcin Pilipczuk\thanks{Faculty of Mathematics, Informatics and Mechanics, University of Warsaw, Poland,
  \texttt{malcin@mimuw.edu.pl}}
  \and
  Magnus Wahlstr{\"o}m\thanks{Department of Computer Science, Royal Holloway, University of London, UK, \texttt{Magnus.Wahlstrom@rhul.ac.uk}}
}

\begin{document}

\maketitle

\begin{abstract}
The input in the \textsc{Minimum-Cost Constraint Satisfaction 
Problem (MinCSP)} over the Point Algebra
contains a set of variables, 
a collection of constraints of the form
$x < y$, $x = y$, $x \leq y$ and $x \neq y$, 
and a budget $k$.
The goal is to check whether it is possible to assign
rational values to the variables while breaking
constraints of total cost at most $k$.
This problem generalizes several prominent
graph separation and transversal problems:
\begin{itemize}
  \item $\textsc{MinCSP}{<}$ is equivalent to \textsc{Directed Feedback Arc Set},
  \item $\textsc{MinCSP}{<,\leq}$ is equivalent to \textsc{Directed Subset Feedback Arc Set},
  \item $\textsc{MinCSP}{=,\neq}$ is equivalent to \textsc{Edge Multicut}, and
  \item $\textsc{MinCSP}{\leq,\neq}$ is equivalent to \textsc{Directed Symmetric Multicut}.
\end{itemize}
Apart from trivial cases, 
$\textsc{MinCSP}{\Gamma}$ for $\Gamma \subseteq \{<,=,\leq,\neq\}$ is 
NP-hard even to approximate within any constant factor under
the Unique Games Conjecture.
Hence, we study parameterized complexity of this problem under a natural parameterization by the solution cost $k$.
We obtain a complete classification:
if $\Gamma \subseteq \{<,=,\leq,\neq\}$ contains both $\leq$ and $\neq$,
then $\textsc{MinCSP}{\Gamma}$ is W[1]-hard, otherwise
it is fixed-parameter tractable.
For the positive cases, we solve $\textsc{MinCSP}{<,=,\neq}$, 
generalizing the FPT results for
\textsc{Directed Feedback Arc Set} and
\textsc{Edge Multicut}
as well as their weighted versions.
Our algorithm works by reducing the problem into a \textsc{Boolean MinCSP},
which is in turn solved by flow augmentation.
For the lower bounds, we prove that
\textsc{Directed Symmetric Multicut} is W[1]-hard,
solving an open problem.
\end{abstract}

\section{Introduction}

The study of graph transversal and separation problems parameterized by solution size is a central research direction in parameterized complexity.
It is usually easy to obtain algorithms for such problems running in $n^{O(k)}$ time, where $n$ is the size of the graph, by enumerating all possible solutions of size at most $k$.
However, obtaining or ruling out \emph{fixed-parameter tractable (FPT)} algorithms, i.e. those running in $f(k) \cdot n^{O(1)}$ time for some computable function $f$ that depends solely on the parameter, is a more challenging task.
Successful examples include the $3^k \cdot n^{O(1)}$ algorithm of Reed, Smith and Vetta~\cite{reed2004finding} for \textsc{Odd Cycle Transversal}, which introduced \emph{iterative compression} -- a technique that has since become a standard opening in FPT algorithms (see e.g. \cite[Chapter~4]{PAbook}).
The algorithms of Marx~\cite{marx2006parameterized} for \textsc{Multiway Cut} and Marx~and~Razgon~\cite{marx2014fixed} for \textsc{Multicut} have introduced, respectively, \emph{important separators} and \emph{shadow removal} to the toolbox of parameterized algorithms; these techniques have also found numerous applications in the field.
Another notable example is the algorithm of Chen~et~al.\ \cite{chen2008fixed} for \textsc{Directed Feedback Arc Set (DFAS)}.

The \textsc{Minimum-Cost Constraint Satisfaction Problem (MinCSP)} provides a unifying framework for modeling optimization problems, including transversal and separator problems in graphs.
The input to a \textsc{MinCSP} is a collection of constraints applied to a set of variables, and 
the goal is to find a solution of minimum cost,
i.e. an assignment of values from a fixed domain to the variables that breaks the minimum number of constraints.
One can cast many problems as \textsc{MinCSP}s by restricting the \emph{constraint language}, i.e. the types of allowed constraints.
More formally, let $\Gamma$ be a set of finitary relations on a fixed domain $D$.
Then $\mincsp{\Gamma}$ is the problem with constraints of the form $R(x_1,\dots,x_r)$, where $R \in \Gamma$ is a relation of arity $r$ and
$(x_1, \dots, x_r) \in V^r$ is a tuple of variables from $V$.
An assignment $\alpha : V \to D$ \emph{satisfies} the constraint $R(x_1, \dots, x_r)$ if $(\alpha(x_1),\dots,\alpha(x_r)) \in R$.

For example, $\mincsp{<}$ on domain $\Rat$ is equivalent to the \textsc{Directed Feedback Arc Set} (DFAS) problem that asks to find a minimum-size set of arcs in a directed graph meeting every cycle.
In other words, deleting this set of arcs makes the graph acyclic.
The reductions in both directions are straightforward: arcs $uv$ in the graph translate into constraints $u < v$, and vice versa.
Clearly, to make a set of $<$-constraints satisfiable, it is necessary and sufficient to delete all cycles of $<$-constraints.
Another example is $\mincsp{=,\neq}$ on domain $\NN$ (or $\QQ$).
This problem is essentially equivalent to the \textsc{Edge Multicut} problem defined as follows: given an undirected graph and vertex pairs $\{s_1, t_1\}, \dots, \{s_m, t_m\}$ called \emph{cut requests}, find a minimum edge set that separates $s_i$ and $t_i$ for all $i$.
Without loss of generality, we may also assume that cut requests are deletable at unit cost~(see e.g.~\cite{marx2014fixed}).
This make the reductions rather simple: edges $uv$ of the graph translate into constraints $u = v$, while cut requests $\{s_i, t_i\}$ translate into constraints $s_i \neq t_i$, and vice versa.
To make a set of such constraints satisfiable, it is necessary and sufficient to ensure that, for every constraint $s \neq t$, there is no $=$-path connecting $s$ and $t$.

As evident from the examples above, many important \textsc{MinCSP}s are NP-hard, and to cope with this, it is natural to parameterize them by solution cost.
This line of work has recently gained momentum,
notably after Kim~et~al.\ introduced
\emph{flow augmentation}~\cite{kim2021solving,kim2022directed} and
successfully used it to resolve the complexity
of $\mincsp{\Gamma}$ for every Boolean constraint
language $\Gamma$, i.e. $\Gamma$ with domain $\{0,1\}$~\cite{kim2023flowIIIsoda}.
Previously, Bonnet~et~al.\ \cite{Bonnet:etal:esa2016,bonnet2016fixed}
classified the complexity of fpt-approximating
\textsc{Boolean MinCSP}s within a constant factor.
Osipov and Wahlstr{\"o}m~\cite{osipov2023parameterized}
resolved parameterized complexity of
exactly solving and constant-factor approximation
of \textsc{Equality MinCSP}s, i.e.
$\mincsp{\Gamma}$ for all $\Gamma$ on domain $\NN$
with relations first-order definable using predicate $=$.

In this paper we consider \textsc{MinCSP} over subsets of \emph{Point Algebra}~\cite{Vilain:Kautz:aaai86,vilain1990constraint}, i.e. $\mincsp{\Gamma}$ for constraint languages $\Gamma \subseteq \{\PA\}$ on domain $\Rat$.
Our motivation is twofold:

\textbf{On the CSP side,} 
constraint language $\{\PA\}$ arguably contains the most basic relations over $\QQ$ and understanding the complexity of $\mincsp{\Gamma}$ for all $\Gamma \subseteq \{\PA\}$ has been identified~\cite{kim2024weighted} as a necessary stepping stone towards broader classification projects, e.g. for all temporal~\cite{bodirsky2010complexity} and interval constraint languages~\cite{Allen:cacm83,Krokhin:etal:jacm2003,bodirsky2022complexity}.
Additionally, Point Algebra is a prominent temporal reasoning formalism in artificial intelligence, and $\mincsp{\PA}$ provides a natural way of dealing with inconsistencies in knowledge bases encoded using this language.
The problem being NP-hard even to approximate within any constant factor (under the Unique Games Conjecture~\cite{khot2002power,svensson2012hardness,guruswami2016simple}) motivates studying parameterized algorithms.
For the study of exact exponential-time algorithms and
polynomial-time approximation, see~\cite{Iwata:Yoshida:stacs2013}.

\textbf{On the FPT side,} the set of problems
$\mincsp{\Gamma}$ for $\Gamma \subseteq \{\PA\}$ 
contains two classical NP-hard graph separation and transversal problems -- DFAS and \textsc{Edge Multicut} -- as well as several robust generalizations and variants.
Studying some of these problems has played an important
role in the development of the parameterized algorithms, and looking at a broader unifying class 
allows exploring the power and limits of existing techniques.
Among the problems within our scope are $\mincsp{<,\leq}$, which is equivalent to \textsc{Directed Subset Feedback Arc Set (Subset-DFAS)}, and $\mincsp{\leq,\neq}$, which is equivalent to \textsc{Directed Symmetric Multicut (DSMC)}.

\textsc{Subset-DFAS} is a variant of DFAS in which the input graph comes with a subset of special arcs, and the goal is to find a minimum transversal for the family of cycles that contain at least one special arc (special arcs translate into $<$-constraints while other arcs -- into $\leq$-constraints, and vice versa).
Parameterized complexity of this problem was resolved by
Chitnis~et~al.\ \cite{chitnis2015directed}
by generalizing the shadow removal technique of~\cite{marx2014fixed} to directed graphs.
Recently, Kim~et~al.\ \cite{kim2024weighted}
gave a flow-augmentation based algorithm 
for this problem that can also handle the arc-weighted version
where the weight budget is not part of the parameter.

In DSMC, we are given a directed graph $G$ and cut requests $\{s_1, t_1\}, \dots, \{s_m, t_m\}$, and the goal is to find a minimum subset of arcs that separates $s_i$ and $t_i$ for all $i$ into distinct strongly connected components.
Again, without loss of generality, we may assume that
a cut request can be ignored at unit cost.
The the translation into $\mincsp{\leq, \neq}$ is as follows:
an arc $uv$ in the graph becomes a constraint $u \leq v$ and
a cut request $\{s_i, t_i\}$ becomes a constraint $s_i \neq t_i$, and vice versa.
Note that a set of $\leq$-constraints is always satisfiable,
and adding a constraint $s_i \neq t_i$ makes it unsatisfiable
if and only if $s_i$ reaches $t_i$ and $t_i$ reaches $s_i$
by directed paths of $\leq$-constraints.
Using the CSP language, it is easy to see that $\mincsp{\leq, \neq}$ generalizes both \textsc{Subset-DFAS} and \textsc{Edge Multicut}: one obtains a cost-preserving reduction from $\mincsp{<,\leq}$ to $\mincsp{\leq,\neq}$ by replacing every constraint of the form $x < y$ with a pair $x \leq y, x \neq y$ and a cost-preserving reduction from $\mincsp{=,\neq}$ to $\mincsp{\leq,\neq}$ by replacing every constraint of the form $x = y$ with $x \leq y$, $y \leq x$.\footnote{
  The latter reduction is analogous to saying that \textsc{Edge Multicut} is the same problem as DSMC on bidirected graphs.
} 
In fact, these reductions show that $\mincsp{\PA}$ is equivalent to $\mincsp{\leq,\neq}$ under cost-preserving reductions.
While \textsc{Edge Multicut} and \textsc{Subset-DFAS} are known to be in FPT (by~\cite{marx2014fixed,bousquet2018multicut} and \cite{chitnis2015directed}, respectively), the parameterized complexity status of DSMC was open~\cite{EibenRW22ipec} (see also~\cite{dabrowski2023parameterized,osipov2023parameterized,kim2024weighted}) prior to our work.
In a related problem called \textsc{Directed Multicut}
the input contains a directed graphs and a set of ordered vertex pairs $\{(s_1, t_1), \dots (s_m, t_m)\}$, 
and the goal is to delete
the smallest number of arcs from the graph so that
no $s_i$ reaches $t_i$.
This problem is W[1]-hard~\cite{marx2014fixed}
even with four pairs (i.e. $m=4$)~\cite{pilipczuk2018directed}.

The scope of our study also includes less prominent and new problems: $\mincsp{<,=}$ (which appeared in~\cite{dabrowski2023parameterized} and was solved by reduction into $\mincsp{<,\leq}$) and $\mincsp{<,=,\neq}$.
These can be thought of as generalizations of \textsc{Multicut} with asymmetric cut requests:
intuitively, if one thinks of the equality constraints as edges in an undirected graph,
the constraints of the form $x \neq y$ correspond to the usual, symmetric cut requests that require separating $x$ and $y$ into distinct connected components,
while constraints of the form $x < y$ additionally require that the connected components can be ordered so that the component of $x$ precedes the component of $y$.

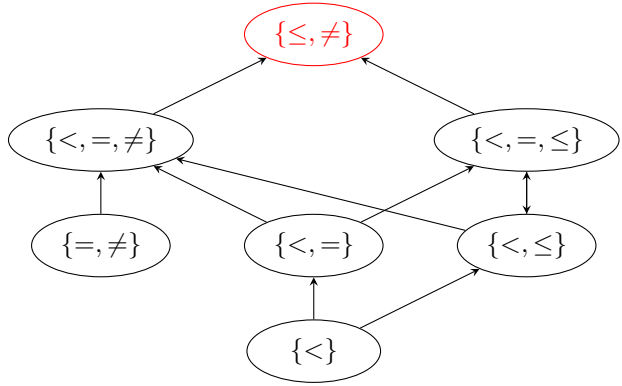
\begin{figure}[tb]
  \centering
  \begin{tikzpicture}[scale=0.7]
  \tikzstyle{every node}=[ellipse, draw=black, minimum width=50pt, align=center]

  \node      at (0,  -6)  (a) {$\{<\}$};
  \node      at (0,  -4)  (b) {$\{<,=\}$};
  \node      at (4,  -4)  (c) {$\{<,\leq\}$};  
  \node      at (-4, -4)  (d) {$\{=,\neq\}$};
  \node      at (4,  -2)  (e) {$\{<,=,\leq\}$};
  \node      at (-4, -2)  (f) {$\{<,=,\neq\}$};
  \node[red] at (0,   0) (g) {$\{\leq,\neq\}$};

  \draw[-stealth] (a) edge (b);
  \draw[-stealth] (a) edge (c);
  \draw[-stealth] (d) edge (f);
  \draw[-stealth] (c) edge (f);
  \draw[-stealth] (c) edge (e);
  \draw[-stealth] (e) edge (c);
  \draw[-stealth] (b) edge (e);
  \draw[-stealth] (b) edge (f);
  \draw[-stealth] (e) edge (g);
  \draw[-stealth] (f) edge (g);
\end{tikzpicture}
  \caption{
    These are seven subsets of Point Algebra.
    Arrow represent polynomial-time cost-preserving
    reductions between corresponding 
    \textsc{MinCSP}s that either follow 
    by inclusion or using
    $(x = y) \equiv (x \leq y) \land (y \leq x)$
    and
    $(x < y) \equiv (x \leq y) \land (x \neq y)$.
    \textsc{MinCSP}s for all of them are NP-hard,
    and in FPT for all except the red language,
    for which it is W[1]-hard.
    There are eight more non-empty subsets,
    out of which four give rise to
    polynomial-time solvable \textsc{MinCSP}s 
    ($\neq$ and $\{\leq,=\}$ with subsets),
    subset $\{<,\neq\}$, which reduced to $\{<\}$
    because all $\neq$-constraints can be safely disregarded,
    and three more 
    subsets that contain both $\leq$ and $\neq$.
  }
  \label{fig:languages}
\end{figure}

\subparagraph*{Our contributions.}
We fully classify the complexity of \textsc{MinCSP} for all subsets of Point Algebra.
For polynomial-time complexity, observe that $\mincsp{\Gamma}$ is NP-hard unless 
$\Gamma \subseteq \{=,\leq\}$ or $\Gamma = \{\neq\}$:
on the one hand, if $\Gamma \subseteq \{=,\leq\}$
or $\Gamma = \{\neq\}$, then every instance is satisfiable
at zero cost using, respectively, 
any constant or any injective assignment;
on the other hand, if $\Gamma$ is not covered by these cases,
then it contains either $\{<\}$, $\{=,\neq\}$ or $\{\leq,\neq\}$
as a subset, all of which imply NP-hardness
by reductions from DFAS and \textsc{Multicut}.
Our main result is the parameterized complexity classification.
We provide algorithms for the weighted version of $\mincsp{\Gamma}$,
where every constraint comes with an integer weight,
and the input also contains a weight budget $W$.
We are allowed to break at most $k$ constraints of total cost at most $W$
-- observe that the parameter is only $k$.
On the other hand, our lower bounds work in the weightless version.

\begin{theorem} \label{thm:classification}
  Let $\Gamma \subseteq \{\PA\}$.
  \begin{enumerate}
    \item If $\Gamma \subseteq \{<, \leq, =\}$,
    then $\wmincsp{\Gamma}$ is fixed-parameter tractable.
    \item If $\Gamma \subseteq \{<, =, \neq\}$,
    then $\wmincsp{\Gamma}$ is fixed-parameter tractable.
    \item Otherwise, $\mincsp{\Gamma}$ is W[1]-hard.
  \end{enumerate}
\end{theorem}

\noindent
The first point is easily obtained by reduction to $\mincsp{<,\leq}$ and applying an FPT algorithm of Chitnis~et~al.\ \cite{chitnis2015directed} or Kim~et~al.\ \cite{kim2024weighted}.
For the second point, we design a new algorithm in Section~\ref{sec:algorithm} using flow augmentation (in the guise of \textsc{Boolean MinCSP}).
Our algorithm for $\mincsp{<,=,\neq}$ combines two strains of flow augmentation-based algorithms that have appeared (explicitly or implicitly) in the literature:
one that is suitable for undirected separation problems like \textsc{Multicut}
(cf.~\cite{kim2021solving,dabrowski2023almost,kim2024weighted})
and another suitable for directed transversal problems like DFAS and \textsc{Subset-DFAS}
(cf.~\cite{kim2022directed,dabrowski2023parameterized,kim2024weighted}).
One could have hoped that pushing these ideas one step further would solve the more general DSMC problem -- alas, this hope is squashed in Section~\ref{sec:w-hardness} where we prove that $\mincsp{\leq, \neq}$, i.e. DSMC, is W[1]-hard.
By noting that every language not covered by the first and the second point of the theorem contains both $\neq$ and $\leq$, this completes the classification.
See Figure~\ref{fig:languages} for an illustration.

\section{Preliminaries}

Fix a domain of values $D$.
A relation $R$ of arity $r$ on $D$ is a subset of tuples in $D^r$.
An instance of a \emph{constraint satisfaction problem (CSP)}
is a set of variables $V$ and a collection of constraints $\cC$
of the form $R(x_1,\dots,x_r)$, where $R$ is a relation of arity $r$ and
$x_1,\dots,x_r$ are variables in $V$.
The set $\{x_1,\dots,x_r\}$ is called the \emph{scope} of constraint $R(x_1,\dots,x_r)$.
An assignment $\alpha : V \to D$ \emph{satisfies a constraint}
$R(x_1,\dots,x_r)$ if $(\alpha(x_1),\dots,\alpha(x_r)) \in R$,
otherwise we say that it \emph{breaks} the constraint.
Similarly, an assignment \emph{satisfies an instance} 
$\cI = (V, \cC)$ of a CSP
if it satisfies all constraints in $\cC$.
A \emph{constraint language} $\Gamma$ is a set of relations on $D$,
and $\csp{\Gamma}$ is the CSP problem in which constraint relations come from the set $\Gamma$.

The \textsc{Minimum-Cost Constraint Satisfaction Problem over $\Gamma$ ($\mincsp{\Gamma}$)} is defined as follows:
given an instance $\cI = (V, \cC)$ of $\csp{\Gamma}$,
a function $\kappa : \cC \to \Int_+$ and
an integer parameter $k$,
decide whether there exists 
exist $X \subseteq \cC$ such that 
$\sum_{C \in X} \kappa(C) \leq k$ and $(V, \cC \setminus X)$ is satisfiable.
We refer to the set $X$ as the solution.
Note that the constraints of cost more than $k$ are satisfied by every feasible assignment,
while a constraint of cost $\ell \leq k$ can be replaced by $\ell$ unit-cost copies.
Thus, we assume without loss of generality that
the cost of every constraint is either $1$ or $\infty$,
and refer to such constraints as \emph{soft} and \emph{crisp}, respectively.
We can also handle an additional weight budget
and arrive at the following definition.

\probDef{$\wmincsp{\Gamma}$}
{Instance $(V, \cC)$ of $\csp{\Gamma}$, 
functions $\kappa : \cC \to \{1,\infty\}$, 
$\omega :\cC \to \Int_+$ and integers $k, W$.}
{$k$.}
{Does there exist $X \subseteq \cC$ such that 
$\sum_{C \in \cC} \kappa(C) \leq k$,
$\sum_{C \in \cC} \omega(C) \leq W$ and 
$(V, \cC \setminus X)$ is satisfiable?}

\noindent 
Note that $W$ is not part of the parameter.
Moreover, no constraint of weight more than $W$
can be part of a solution, so we assume those
are always crisp.

The \emph{Point Algebra} is a constraint language on
the domain $\Rat$ with four binary relations
$<$, $\leq$, $=$ and $\neq$.
We use infix notation for Point Algebra constraints, e.g.
$x < y$ and $x \neq y$.
One can check in polynomial time
whether an instance of $\csp{\PA}$ is satisfiable by, e.g., reducing to $\csp{\neq,\leq}$, constructing
directed graph of $\leq$-constraints and,
for every constraint $x \neq y$, checking that
$x$ and $y$ are not strongly connected
(cf.~\cite{vilain1990constraint}).

\section{FPT Algorithm for $\mincsp{<,=,\neq}$}
\label{sec:algorithm}

In this section we prove the following theorem.

\begin{theorem}
  \label{thm:mincsp-less-eq-neq-fpt}
  $\wmincsp{<,=,\neq}$ is fixed-parameter tractable.
\end{theorem}

Our algorithm follows the compress-branch-cut 
paradigm~(cf. \cite{reed2004finding,chen2008fixed,marx2014fixed}).
On the high level,
we use iterative compression to obtain 
an approximate solution of size
bounded by a function of the parameter.
This allows us to exhaustively guess 
certain information about the variables
appearing in the approximate solution
and then reduce the problem to
a fixed-parameter tractable \textsc{Boolean MinCSP}.
The last step blends two types of reductions suitable for symmetric and asymmetric properties, respectively.

We start this section by describing the compression
and branching steps, which are quite standard.
Then we set up the \textsc{Boolean MinCSP} machinery,
describe the algorithm for
$\mincsp{<,=,\neq}$ and prove its correctness.
We finish with some high-level observations
about the algorithm and the obstacles in
pushing it further.

\paragraph*{Compression and Branching}

Let $\cI = (V, \cC, \kappa, \omega, k, W)$ be an instance of 
$\wmincsp{<,=,\neq}$.
By a standard iterative compression argument~(cf. \cite[Chapter~4]{PAbook}), 
we may assume access to a set of constraints $X_{\rm in} \subseteq \cC$
such that $(V, \cC \setminus X_{\rm in})$ is satisfiable.
Moreover, if $\cI$ is a yes-instance,
we may assume that it admits an optimal solution 
$X_{\rm opt}$ disjoint from $X_{\rm in}$.
Since $|V(X_{\rm in})| \leq 2(k+1)$,
we can enumerate ordered partitions of $V(X_{\rm in})$
in $O^*(k^k)$ time and thus guess an assignment 
$\alpha_{\rm in} : V(X_{\rm in}) \to \Rat$
that \emph{agrees} with $\alpha_{\rm out}$ in the following sense:
for every $x, y \in V(X_{\rm in})$, we have
\begin{itemize}
  \item $\alpha_{\rm in}(x) = \alpha_{\rm in}(y)$ if and only if $\alpha_{\rm out}(x) = \alpha_{\rm out}(y)$, and
  \item $\alpha_{\rm in}(x) < \alpha_{\rm in}(y)$ if and only if $\alpha_{\rm out}(x) < \alpha_{\rm out}(y)$.
\end{itemize}
Since $X_{\rm in} \cap X_{\rm out} = \emptyset$ and
$\alpha_{\rm in}$ agrees with $\alpha_{\rm opt}$,
we obtain that $\alpha_{\rm in}$ satisfies $X_{\rm in}$.
Moreover, every assignment that agrees with $\alpha_{\rm in}$
satisfies $X_{\rm in}$ as well, so
we can safely remove $X_{\rm in}$ from $\cC$.
Assuming our guesses were correct,
we now have a satisfiable instance 
$(V, \cC \setminus X_{\rm in})$ of $\csp{<,=,\neq}$
and, additionally, a partial assignment
$\alpha_{\rm in} : V(X_{\rm in}) \to \Rat$
with the promise that there is an assignment to $\cI$
that agrees with $\alpha_{\rm in}$
and breaks constraints of cost at most $k$ and weight at most $W$.
Moreover, if $\alpha_{\rm in}(x) = \alpha_{\rm in}(y)$
for some $x,y \in V(X_{\rm in})$, then we can identify $x$ and $y$.
Thus, in time $O^*(k^k)$
we have reduced the problem to the following version.

\probDef{$\cwmincsp{<,=,\neq}$}
{A satisfiable instance $(V, \cC)$ of $\csp{<,=,\neq}$,
functions $\kappa : \cC \to \{1,\infty\}$ and $\omega :\cC \to \Int_+$,
integers $k$ and $W$,
a subset $U \subseteq V$ with $|U| \leq 2(k+1)$
and an injective assignment $\alpha : U \to \Rat$.}
{$k$.}
{Does there exist $X \subseteq \cC$ such that
$\sum_{C \in X} \kappa(C) \leq k$,
$\sum_{C \in X} \omega(C) \leq W$, and
$(V, \cC \setminus X)$ is satisfiable by an assignment
that agrees with $\alpha$?}

By a standard correctness argument~(cf. \cite[Chapter~4]{PAbook}),
we have the following.

\begin{proposition}
  \label{prop:mincsp-compress}
  If $\cwmincsp{<,=,\neq}$ is fixed-parameter tractable,
  then \\
  $\wmincsp{<,=,\neq}$ is fixed-parameter tractable.
\end{proposition}

To solve the compressed version,
we reduce it into a separation problem
handled by the \textsc{Boolean MinCSP} machinery.

\paragraph*{\textsc{Cutting tool: Boolean MinCSP}}

A Boolean constraint language is a set of 
relations on domain $\{0,1\}$.
Kim~et~al.~\cite{kim2023flowIIIsoda}
classified parameterized complexity
of $\wmincsp{\Gamma}$ for every 
finite Boolean constraint language $\Gamma$.
We will use the positive part of their classification.
To describe it, we start with some definitions.

A Boolean relation $R \subseteq \{0,1\}^r$ can be modeled
as a set of satisfying assignments to a propositional formula on $r$ variables.
Formally, let $\phi$ be a propositional formula on
variables $x_1, \dots, x_r$.
For $b_1,\dots,b_r \in \{0,1\}$, let $\phi(b_1, \dots, b_r)$
be the Boolean expression obtained by substituting
variable $x_i$ with the Boolean value $b_i$ for all $i \in \{1,\dots,r\}$.
Then, for every $R \subseteq \{0,1\}^r$, there exists a formula $\phi$ such that
\[ R = \{ (b_1,\dots,b_r) \in \{0,1\}^r \mid \phi(b_1,\dots,b_r) \text{ evaluates to true} \}. \] 
We remark in the passing that $\phi$ does not have to be unique.

A propositional formula is \emph{bijunctive} if it is a conjunction
of $2$-clauses, i.e. sub-formulas of the form
$(1 \to x)$, $(x \to 0)$, $(x \to y)$, $(x \lor y)$ and $(\neg x \lor \neg y)$.
Associate an undirected \emph{Gaifman graph} $G_\phi$ with $\phi$:
let $\{1,\dots,r\}$ be the vertices of $G_\phi$ and
let $G_\phi$ contain an edge $ij$ whenever 
$x_i$ and $x_j$ appear together in a $2$-clause in $\phi$.
We say that a Boolean relation is \emph{bijunctive}
if it is definable by a bijunctive propositional formula.
Furthermore, a bijunctive relation is \emph{$2K_2$-free} if
it is definable by a bijunctive formula $\phi$ such that $G_\phi$
does not contain an induced $2K_2$, i.e. 
for every pair of vertex-disjoint edges $ij$ and $i'j'$ in $G_\phi$, 
there is an edge with one endpoint in $\{i,j\}$ 
and another in $\{i',j'\}$. 

\begin{theorem}[Theorem~1.2~in~\cite{kim2023flowIIIsoda}]
  \label{thm:boolean-mincsp-fpt}
  Let $\Gamma$ be a set of $2K_2$-free bijunctive 
  Boolean relations of arity at most $r$.
  Then $\wmincsp{\Gamma}$ is in FPT parameterized by $k + r$.
\end{theorem}

In the forthcoming description of the algorithm
we will abuse the notation and define Boolean constraints 
explicitly using propositional formulas.

\paragraph*{Solving the Compressed Version}

We define a procedure that takes an instance $\cI$ of \textsc{Compressed}
$\wmincsp{<,=,\neq}$ 
as input
and in polynomial time outputs an instance $\cI'$
of \textsc{Weighted} $\mincsp{\Gamma}$ with the same parameter $k$, 
with $\Gamma$ being a Boolean constraint language.
We will argue that
all relations in $\Gamma$ are bijunctive $2K_2$-free 
and of arity at most $O(k)$ (in \Cref{lem:2k2-free}),
and 
$\cI$ is a yes-instance if and only if $\cI'$ is 
a yes-instance (in \Cref{lem:reduce-sound,lem:reduce-complete}).
Combined with \Cref{prop:mincsp-compress} and \Cref{thm:boolean-mincsp-fpt},
this yields \Cref{thm:mincsp-less-eq-neq-fpt}.

We start by describing the reduction procedure.
Consider an instance $\cI = (V, \cC, \kappa, \omega, k, W, U, \alpha)$
of $\cwmincsp{<,=,\neq}$.
Construct the output instance 
$\cI' = (V', \cC', \kappa', \omega', k, W)$ of $\mincsp{\Gamma}$ 
with the same values $k$ and $W$ as follows.
Let $\ell = |U|$ and observe that $\ell \leq 2(k+1)$.
Enumerate $U$ as $u_1, \dots, u_\ell$ ordered so that
$\alpha(u_1) < \alpha(u_2) < \dots < \alpha(u_\ell)$.
For every variable $v \in V$, introduce
Boolean variables
$c_{v,i}$ for $i \in \{1,\dots,\ell\}$ and
$p_{v,j}$ for $j \in \{1,\dots,2\ell+1\}$ in $V'$.

Before defining the constraints of $\cC'$,
we elaborate on the intended interpretation
of the Boolean variables in $V'$.
Suppose $\alpha_{\rm opt}$ is an assignment
that breaks constraints of total cost at most $k$,
total weight at most $W$ and 
agrees with $\alpha_{\rm in}$ on $U$.
Variables $c_{v,i}$ indicate whether
$\alpha_{\rm opt}(v)$ equals $\alpha_{\rm opt}(u_i)$ 
for some $u_i \in U$.
Since $\alpha_{\rm in}$ is injective,
$c_{v,i}$ equals $1$ for at most one $i$.
Variables $p_{v,j}$ encode how the values
$\alpha_{\rm opt}(v)$ are ordered with 
respect to $\alpha_{\rm opt}(u_i)$.
More precisely, $p_{v,2i}$ is set to $1$ if and only if 
$\alpha_{\rm opt}(v) \geq \alpha_{\rm opt}(u_{i})$ and
$p_{v,2i+1}$ is set to $1$ if and only if
$\alpha_{\rm opt}(v) > \alpha_{\rm opt}(u_{i})$.
Note that $p_{v,1} = 1$ holds vacuously,
while $p_{v,j'}=1$ implies $p_{v,j} = 1$ for all $j < j'$
since $\alpha_{\rm in}(u_j) < \alpha_{\rm in}(u_{j'})$.

With the intuition in mind, 
we populate $\cC'$ with crisp constraints.
For every $v \in V$ and every $1 \leq i < i' \leq \ell$, 
add a crisp Boolean constraint
\begin{equation}
  \label{eq:at-most-one}
  (\neg c_{v,i} \lor \neg c_{v,i'}).
\end{equation}
For every $u_i \in U$,
add a crisp Boolean constraint
\begin{equation}
  \label{eq:ui-cc}
  (1 \to c_{u_i,i}).
\end{equation}
For every $v \in V$,
add crisp Boolean constraints
\begin{equation}
  \label{eq:chain}
  (1 \to p_{v,1}) \text{ and }
  (p_{v,j} \leftarrow p_{v,j'}) \text{ for all } 1 \leq j < j' \leq 2\ell+1.
\end{equation}
For every $u_i \in U$,
add crisp Boolean constraints
\begin{equation}
  \label{eq:ui-ord}
  (1 \to p_{u_i, 2i}), (p_{u_i, 2i+1} \to 0).
\end{equation}
Finally, for every $v \in V$ and $i \in \{1,\dots,\ell\}$,
add crisp Boolean constraints
\begin{equation}
  \label{eq:c->p}
  \begin{aligned}
    &(c_{v,i} \to p_{v,j})      &&\text{ for all } 1 \leq j \leq 2i, \text{ and} \\
    &(c_{v,i} \to \neg p_{v,j}) &&\text{ for all } 2i < j \leq 2\ell+1.
  \end{aligned}
\end{equation}

Now, for every $C \in \cC$, we define a constraint $C'$ in $\cC'$
with $\kappa'(C') = \kappa(C)$ and $\omega'(C') = \omega(C)$.
If $C$ is an equality constraint $(v = w)$,
then let $C'$ be the Boolean constraint
\begin{equation}
  \label{eq:equality}
  \begin{aligned}
    \bigwedge_{i=1}^{\ell} (c_{v,i} = c_{w,i}) &\land 
    \textcolor{purple}{
      \bigwedge_{i=1}^{\ell} \bigwedge_{i'=i+1}^{\ell} (\neg c_{v,i} \lor \neg c_{v,i'})
    } \land \\
    \bigwedge_{j=1}^{2\ell+1} (p_{v,j} = p_{w,j}) &\land
    \textcolor{purple}{
      \bigwedge_{j=1}^{2\ell+1} \bigwedge_{j'=j+1}^{2\ell+1} (p_{v,j} \leftarrow p_{v,j'})
    } \land \\
    \textcolor{purple}{
      \bigwedge_{i=1}^{\ell} \bigwedge_{j=1}^{2i} (c_{v,i} \to p_{v,j})
    } &\land
    \textcolor{purple}{
      \bigwedge_{i=1}^{\ell} \bigwedge_{j=2i+1}^{2\ell+1} (c_{v,i} \to \neg p_{v,j})
    },
  \end{aligned}
\end{equation}
where clause $(x = y)$ is a short-hand for 
the conjunction $(x \to y) \land (x \leftarrow y)$.
If $C$ is an disequality constraint $(v \neq w)$,
then let $C'$ be the Boolean constraint
\begin{equation}
  \label{eq:disequality}
  \bigwedge_{i=1}^{\ell} (\neg c_{v,i} \lor \neg c_{w,i}) \land 
  \textcolor{purple}{
    \bigwedge_{i=1}^{\ell} \bigwedge_{i'=i+1}^{\ell} (\neg c_{v,i} \lor \neg c_{v,i'})
  }
\end{equation}
Finally, 
if $C$ is an ordering constraint $(v < w)$,
then let $C'$ be the Boolean constraint
\begin{equation}
  \label{eq:ordering}
  \bigwedge_{i=1}^{\ell} (p_{v,2i-1} \to p_{w,2i-1}) \land 
  \bigwedge_{i=1}^{\ell} (p_{v,2i} \to p_{w,2i+1}) \land
  \textcolor{purple}{
    \bigwedge_{j=1}^{2\ell+1} \bigwedge_{j'=j+1}^{2\ell+1} 
    (p_{v,j} \leftarrow p_{v,j'})
  }
\end{equation}
This completes the construction.

Clearly, the reduction requires polynomial time.
To show that $\cI'$ can be decided in FPT time with respect to $k$,
we use \Cref{thm:boolean-mincsp-fpt} and the following lemma.

\begin{lemma}
  \label{lem:2k2-free}
  Every Boolean relation used in the constraints
  of $\cC'$ is bijunctive $2K_2$-free and
  of arity at most $8k + 10$.
\end{lemma}
\begin{proof}
  The Gaifman graphs of the propositional formulas
  in~\Cref{eq:at-most-one,eq:ui-cc,eq:chain,eq:ui-ord,eq:c->p}
  contain at most two vertices.
  The Gaifman graphs of propositional formulas
  in~\Cref{eq:equality,eq:disequality,eq:ordering} 
  consist of a clique (formed by the edges corresponding to
  clauses that also appear as crisp constraints defined
  in~\Cref{eq:at-most-one,eq:chain,eq:c->p} and are
  highlighted in purple) 
  and edges with exactly one endpoint in the clique.
  These graphs are $2K_2$-free.
  The constraints of maximum arity are defined 
  in~\Cref{eq:equality}, each contains
  $2(\ell+1 + \ell) = 4\ell + 2 \leq 4(2k+2) +2 \leq 8k + 10$
  variables in its scope.
\end{proof}

To show that the reduction is correct,
we start with the forward direction.

\begin{lemma}
  \label{lem:reduce-sound}
  If $\cI$ is a yes-instance of $\cwmincsp{<,=,\neq}$, then
  $\cI'$ is a yes-instance of $\wmincsp{\Gamma}$.
\end{lemma}
\begin{proof}
  Let $X$ be a solution to $\cI$, i.e.
  a subset of constraints in $\cC$ such that
  $\sum_{C \in X} \kappa(C) \leq k$,
  $\sum_{C \in X} \omega(C) \leq W$, and
  $(V, \cC \setminus X)$ admits
  a satisfying assignment $\beta : V \to \Rat$
  that agrees with the partial assignment $\alpha : U \to \Rat$.
  Define a subset of Boolean constraints
  $X' = \{ C' \in \cC' : C \in X \}$.
  By construction, $X'$ has cost at most $k$
  and weight at most $W$,
  so it remains to show that 
  $(V', \cC' \setminus X')$ is satisfiable.
  To this end, we define an assignment
  $\beta' : V' \to \{0,1\}$ as follows.
  For every $v \in V$ and $i \in \{1,\dots,\ell\}$,
  set 
  $\beta'(c_{v,i}) = 1$ if and only if $\beta(c_{v,i}) = \beta(u_i)$.
  For every $v \in V$ and $i \in \{1,\dots,\ell\}$,
  set 
  $\beta'(p_{v,1}) = 1$,
  $\beta'(p_{v,2i}) = 1$ if and only if $\beta(v) \geq \beta(u_i)$, and
  $\beta'(p_{v,2i+1}) = 1$ if and only if $\beta(v) > \beta(u_i)$.
  
  Observe that since $\alpha$ is injective
  and $\beta$ agrees with $\alpha$,
  the assignment $\beta$ is also injective on $U$,
  so $\beta'$ satisfies the constraints defined in~\Cref{eq:at-most-one,eq:ui-cc,eq:ui-ord}.
  Further, $(\beta'(p_{v,1}), \dots, \beta'(p_{v,2\ell+1}))$
  is a vector starting with ones followed by (possibly no) zeros,
  so $\beta'$ satisfies the constraints defined in~\Cref{eq:chain}.
  Moreover, $\beta'(c_{v,i}) = 1$ if and only if $\beta(v) = \beta(u_i)$,
  in which case $\beta'(p_{v,2i}) = 1$ and $\beta'(p_{v,2i+1}) = 0$,
  hence $\beta'$ satisfies the constraints defined in~\Cref{eq:c->p}.

  Now consider a constraint $C \in \cC \setminus X$ satisfied by $\beta$.
  We claim that $\beta'$ satisfies the corresponding 
  Boolean constraint $C' \in \cC \setminus X'$.
  Note that it is sufficient to check that 
  the clauses not present in the crisp constraints defined 
  in~\Cref{eq:at-most-one,eq:chain,eq:c->p} are satisfied.
  Suppose $C$ is an equality constraint $(v = w)$
  and recall the definition of $C'$ from~\Cref{eq:equality}.
  By definition of $\beta'$, it is clear that
  $\beta(v) = \beta(w)$ implies that
  $\beta'(c_{v,i}) = \beta'(c_{w,i})$ for all $i$ and
  $\beta'(p_{v,j}) = \beta'(p_{w,j})$ for all $j$,
  so $\beta'$ satisfies $C'$.
  Now suppose $C$ is a disequality constraint $(v \neq w)$ and
  recall the definition of $C'$ from~\Cref{eq:disequality}.
  If $\beta(v) = \beta(u_i)$ for some $i$, then 
  $\beta(w) \neq \beta(u_i)$,
  and $\beta'(c_{v,i}) = 1$ implies that $\beta'(c_{w,i}) = 0$.
  Otherwise, $\beta(v) \neq \beta(u_i)$ for all $i$,
  and $\beta'(c_{v,i}) = 0$ for all $i$.
  In both cases, $\beta'$ satisfies $C'$.
  Finally, if $C$ is an ordering constraint $(v < w)$,
  recall the definition of $C'$ from~\Cref{eq:ordering}.
  If $\beta(v) = \beta(u_i)$ for some $i$,
  then $\beta(u_i) < \beta(w)$,
  hence $\beta'(p_{v,2i}) = 1$, 
  $\beta'(p_{v,2i+1}) = 0$ and
  $\beta'(p_{w,2i}) = \beta'(p_{w,2i+1}) = 1$.
  Otherwise, we have three more cases:
  \begin{itemize}
    \item 
    if $\beta(v) < \beta(u_1)$, then
    $\beta'(p_{v,j}) = 0$ for all $j \geq 2$;
    \item 
    if $\beta(u_{i-1}) < \beta(v) < \beta(u_{i})$ 
    for some $i$, then $\beta(u_{i-1}) < \beta(w)$, so
    $\beta'(p_{v,j}) = 0$ for all $j \geq 2i$ and
    $\beta'(p_{w,j'}) = 1$ for all $j' < 2i$;
    \item 
    if $\beta(v) > \beta(u_\ell)$, then 
    $\beta(w) > \beta(u_\ell)$,
    so $\beta'(p_{v,j}) = \beta'(p_{w,j}) = 1$
    for all $j$.
  \end{itemize}
  In all cases, $\beta'$ satisfies $C'$.
\end{proof}

To complete the proof of \Cref{thm:mincsp-less-eq-neq-fpt},
it remains to show the converse \Cref{lem:reduce-sound}.

\begin{lemma}
  \label{lem:reduce-complete}
  If $\cI'$ is a yes-instance of $\wmincsp{\Gamma}$,
  then $\cI$ is a yes-instance of \textsc{Compressed} $\wmincsp{<,=,\neq}$.
\end{lemma}
\begin{proof}
  Let $X'$ be a solution to $\cI'$, i.e.
  a subset of constraints in $\cC'$ such that
  $\sum_{C' \in X'} \kappa'(C') \leq k$,
  $\sum_{C' \in X'} \omega'(C') \leq W$, and
  $(V', \cC' \setminus X')$ is satisfiable.
  Fix an assignment $\beta' : V' \to \{0,1\}$
  that satisfies $(V', \cC' \setminus X')$ and
  define $X = \{ C \in \cC : C' \in X' \}$.
  By construction, we have
  $\sum_{C \in X} \kappa(C) \leq k$ and
  $\sum_{C \in X} \omega(C) \leq W$,
  so it remains to show that 
  $(V, \cC \setminus X)$ admits a satisfying assignment
  $\beta : V \to \Rat$ that agrees with $\alpha : U \to \Rat$.
  To define such an assignment,
  recall that $(V, \cC)$ is satisfiable,
  and let $\gamma : V \to \Rat$ be a satisfying assignment.
  Without loss of generality,
  assume that the range of $\gamma$ is $(0,1)$.
  Define $\iota : V \to \{1,\dots,\ell\}$
  as $\iota(v) \coloneqq \max\{i : \beta'(p_{v,2i}) = 1\}$ and let
  \[
    \beta(v) = \begin{cases}
      i & \text{if there exists } i \text{ such that } \beta'(c_{v,i}) = 1, \\
      \iota(v) + \gamma(v) & \text{otherwise}.
    \end{cases}
  \]
  Note that $\beta(u_i) = i$ for all $i$,
  so it agrees with $\alpha$ on $U$.
  We claim that if $\beta'$ satisfies 
  a constraint $C' \in \cC' \setminus X'$, then
  $\beta$ satisfies the corresponding 
  constraint $C \in \cC \setminus X$.
  Note that by the crisp constraints from~\Cref{eq:at-most-one},
  we have $\beta'(c_{v,i}) = 1$ for at most one $i$.
  Suppose $C$ is an equality constraint $(v = w)$,
  and recall $C'$ from~\Cref{eq:equality}.
  If there exists $i$ such that $\beta'(c_{v,i}) = 1$,
  then $\beta'(c_{w,i}) = 1$ and
  $\beta(v) = \beta(w) = i$.
  Otherwise, note that $\iota(v) = \iota(w)$
  since $\beta'(p_{v,j}) = \beta'(p_{w,j})$ for all $j$,
  and $\gamma(v) = \gamma(w)$ since $\gamma$ satisfies $(v = w)$.
  Hence, $\beta(v) = \beta(w)$ in this case as well.
  Now suppose that $C$ is a disequality constraint $(v \neq w)$,
  and recall $C'$ from~\Cref{eq:disequality}.
  If there exists $i$ such that $\beta'(c_{v,i}) = 1$,
  then $\beta'(c_{w,i}) = 0$ and
  $\beta(v) = i \neq \beta(w)$.
  Otherwise, $\gamma(v) \neq \gamma(w)$ since 
  $\gamma$ satisfies $(v \neq w)$,
  so the fractional parts of $\beta(v)$ and $\beta(w)$
  are distinct, and $\beta(v) \neq \beta(w)$.
  Finally, suppose $C$ is an ordering constraint $(v < w)$.
  If there exists $i$ such that $\beta'(c_{v,i}) = 1$,
  then $\beta'(p_{v,2i}) = 1$ and $\beta'(p_{w,2i+1}) = 1$,
  so $\beta(v) = i$ and $\beta(w) = \iota(w) + \gamma(w) > i$.
  Otherwise, since $p_{v,2i-1} \to p_{w,2i-1}$ for all $i$,
  we have $\iota(v) \leq \iota(w)$ and
  $\gamma(v) < \gamma(w)$ since $\gamma$ satisfies $(v < w)$.
  Hence, $\beta(v) < \beta(w)$, as desired.
\end{proof}

\section{W[1]-hardness of \textsc{Directed Symmetric Multicut}}
\label{sec:w-hardness}

Recall the definition of the problem.

\probDef{Directed Symmetric Multicut (DSMC)}
{A directed graph $D$, a collection $\cP$ of vertex pairs in $D$, and integer $k$.}
{$k$.}
{Does there exist a set $X$ of arcs in $D$ such that
no pair $\{s,t\} \in \cP$ is strongly connected in $D - X$?}

Pairs in $\cP$ are referred to as cut requests.
In this section we prove the following.

\begin{theorem}
  \label{thm:dsmc-is-w1hard}
  \textsc{Directed Symmetric Multicut} parameterized
  by solution size is W[1]-hard.
\end{theorem}

Our starting point is a multicolored
variant of \textsc{$k$-Clique}:
given an undirected graph $G$, 
an integer $k$ and
a partition
$V^1 \uplus \dots \uplus V^k$ of $V(G)$
such that $|V^1| = \dots = |V^k| = n$,
the question is whether $G$
contains a complete subgraph 
with one vertex from each $V^i$.
Given an instance 
$(G, k, V^1 \uplus \dots \uplus V^k)$ 
of this problem,
we construct an instance $(D, \cP, k')$ of DSMC,
where $k' = 3k^2$.

\begin{figure}
    \centering
    \begin{tikzpicture}[every node/.style={draw,circle,minimum size=20pt}, scale=0.4]
  \node at (-1.5,0) (w) {w};
  \node at (0,3)    (n) {n};
  \node at (1.5,0)  (e) {e};
  \node at (0,-3)   (s) {s}; 

  \node at (5.5,0)  (w1) {};
  \node at (7,3)    (n1) {};
  \node at (8.5,0)  (e1) {};
  \node at (7,-3)   (s1) {}; 

  \node at (10,3)   (n2) {};
  \node at (11.5,0) (e2) {};
  \node at (10,-3)  (s2) {}; 

  \node at (13,3)   (n3) {};
  \node at (14.5,0) (e3) {};
  \node at (13,-3)  (s3) {};   
  
  \draw[->,thick] (s) -- (w); 
  \draw[->,thick] (w) -- (n);
  \draw[->,thick] (s) -- (e);
  \draw[->,thick] (e) -- (n);
  \draw[->] (n) -- (s);

  \draw[->,thick] (s1) -- (w1); 
  \draw[->,thick] (w1) -- (n1);
  \draw[->,thick] (s1) -- (e1);
  \draw[->,thick] (e1) -- (n1);  
  \draw[->,thick] (s2) -- (e1); 
  \draw[->,thick] (e1) -- (n2);
  \draw[->,thick] (s2) -- (e2);
  \draw[->,thick] (e2) -- (n2);
  \draw[->,thick] (s3) -- (e2); 
  \draw[->,thick] (e2) -- (n3);
  \draw[->,thick] (s3) -- (e3);
  \draw[->,thick] (e3) -- (n3);  
  \draw[->] (n1) -- (s1);
  \draw[->] (n2) -- (s2);
  \draw[->] (n3) -- (s3);
\end{tikzpicture}
    \caption{A diamond digraph on the left and
    a sequence of three joined diamonds on the right. 
    Undeletable arcs are drawn with thick lines,
    and deletable arcs are drawn with a thin line.}
    \label{fig:diamond}
\end{figure}
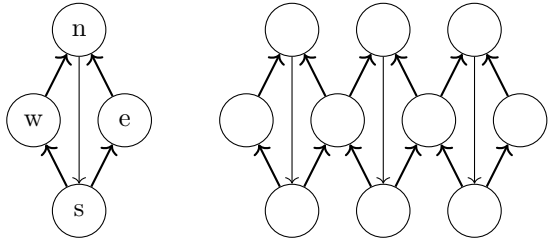

On the high level, the construction
consists of two parts --
choice gadgets
for every vertex set $V^i$, and
coordination gadgets that 
disallow choosing non-adjacent vertices
into the solution.
The building blocks of the choice gadgets
are \emph{diamonds}, which are graphs 
$\lozenge$
on four vertices $w, n, e, s$ 
(for west, north, east and south)
with four undeletable arcs 
$sw$, $wn$, $se$, $en$ and
one deletable arc $ns$.
Note that $\lozenge$ is strongly connected,
but loses this property if $ns$ is deleted.
We refer to vertices $w$ and $e$ of a diamond as \emph{junction} vertices.
By \emph{picking a diamond} we mean deleting the arc $ns$. 
By \emph{joining two diamonds}
we mean identifying eastern vertex of the first
with the western vertex of the other.
Note that the result of sequentially joining diamonds is 
a strongly connected graph,
but loses this property
if any diamond is picked.
See Figure~\ref{fig:diamond} for an illustration.

\paragraph*{Choice gadgets}

Denote the vertices of $V^i$ by $v^i_1, \dots, v^i_n$.
For every subset $V^i$ in $G$,
create a \emph{necklace} of $3kn$ diamonds
in $D$ joined in a cyclic fashion:
construct $k$ \emph{strings} $S^i_j$ for
$1 \leq j \leq k$,
where each $S^i_j$ 
consists of $3n$ sequentially joined diamonds
$\lozenge^{i,j}_1, \dots, \lozenge^{i,j}_{3n}$;
create the necklace by 
joining the last diamond of string
$S^{i}_{j}$ with the first diamond of
$S^{i}_{j+1 \bmod k}$ for all $1 \leq j \leq k$.
Let the junction vertices of the diamonds
on the necklace be
$c_0, \dots, c_{3kn-1}$.
Add cut requests
$\{c_\alpha, c_{\alpha + n \bmod 3kn}\}$
to $\cP$ 
for all $0 \leq \alpha < 3kn$.

Before proceeding with coordination gadgets,
we observe some properties of the necklace.
To satisfy cut requests introduced so far,
a solution needs to pick at least $3k$ diamonds since
we have disjoint requests $\{c_{(i-1)n}, c_{in}\}$ 
for every $i \in [3k]$.
The budget is $k' = k \cdot 3k$,
and there are $k$ such necklaces,
so a solution has to
pick exactly $3k$ diamonds in each necklace.
We claim that picked diamonds are evenly spaced, 
i.e. there exists $\alpha \in \{1,\dots,n\}$ 
such that the solution picks diamonds
$\lozenge^{i,j}_{\alpha}$, $\lozenge^{i,j}_{\alpha + n}$ and 
$\lozenge^{i,j}_{\alpha + 2n}$
in the strings $S^{i}_{j}$ for all $1 \leq j \leq k$.
Indeed, if a hypothetical solution
picks any other set of $3k$ diamonds,
then a sequence of at least $n$
joined and unpicked diamonds remains on the necklace,
and any such sequence contains 
an unsatisfied cut request.
Intuitively, we can interpret solution
picking diamonds $\lozenge^{i,1}_{\alpha}$ as choosing vertex $v^i_\alpha \in V^i$ 
to be part of the clique in $G$.

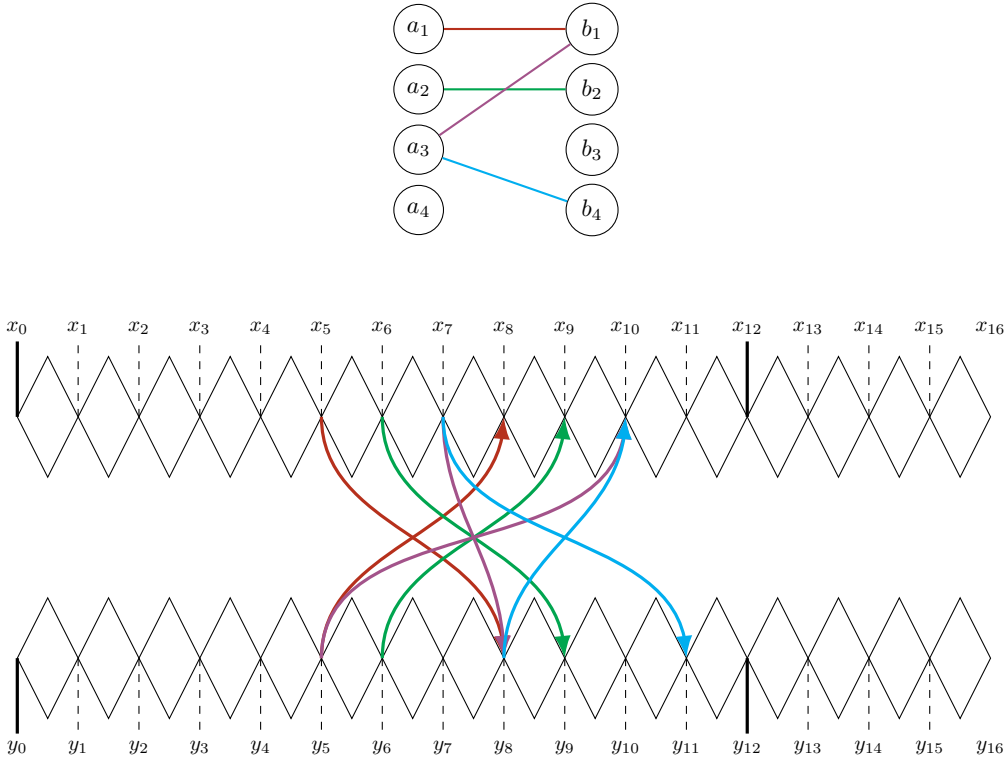
\begin{figure}[tb]
  \centering
  \begin{tikzpicture}[scale=0.4, every node/.style={draw,circle,scale=0.9}]



  \def \offset {0}

  \node at (0+\offset,6) [left] (a1) {$a_1$};
  \node at (0+\offset,4) [left] (a2) {$a_2$};
  \node at (0+\offset,2) [left] (a3) {$a_3$};
  \node at (0+\offset,0) [left] (a4) {$a_4$};

  \node at (4+\offset,6) [right] (b1) {$b_1$};
  \node at (4+\offset,4) [right] (b2) {$b_2$};
  \node at (4+\offset,2) [right] (b3) {$b_3$};
  \node at (4+\offset,0) [right] (b4) {$b_4$};

  \draw[thick,BrickRed]     (a1) -- (b1);
  \draw[thick,Green]        (a2) -- (b2);
  \draw[thick,DarkOrchid]   (a3) -- (b1);
  \draw[thick,Cyan]         (a3) -- (b4);
\end{tikzpicture}
  
  \vspace{1cm}
  
  \begin{tikzpicture}[scale=0.4, every node/.style={scale=0.8}]
  \def \offset {8}
  \def \k {3}
  \def \n {4}

  \foreach \i in {0,1,...,15}{
    \draw (-1+2*\i,0) -- (2*\i,2) -- (1+2*\i,0) -- (2*\i,-2) -- (-1+2*\i,0);
    \draw (-1+2*\i,0-\offset) -- (2*\i,2-\offset) -- (1+2*\i,0-\offset) -- (2*\i,-2-\offset) -- (-1+2*\i,0-\offset);

    \node at (-1+2*\i,2.5) [above] {$x_{\i}$};
    \node at (-1+2*\i,-2.5-\offset) [below] {$y_{\i}$};
    \draw[dashed] (-1+2*\i,0) -- (-1+2*\i,2.5);
    \draw[dashed] (-1+2*\i,0-\offset) -- (-1+2*\i,-2.5-\offset);
  }
  \node at (-1+2*16,2.5) [above] {$x_{16}$};
  \node at (-1+2*16,-2.5-\offset) [below] {$y_{16}$};

  \foreach \i in {0,12}{
    \draw[very thick] (-1+2*\i,0) -- (-1+2*\i,2.5);
    \draw[very thick] (-1+2*\i,0-\offset) -- (-1+2*\i,-2.5-\offset);    
  }

  \def \x {1};
  \def \y {1};
  \def \col {BrickRed};
  
  \draw[-{Latex}, \col, very thick] (-1 + 2 * \n + 2 * \x, 0) to 
  [out=-90, in=90] (-3 + 4 * \n + 2 * \y, -\offset);  
  \draw[-{Latex}, \col, very thick] (-1 + 2 * \n + 2 * \y, -\offset) to 
  [out=90, in=-90] (-3 + 4 * \n + 2 * \x, 0);
  
  \def \x {2};
  \def \y {2};
  \def \col {Green};
  
  \draw[-{Latex}, \col, very thick] (-1 + 2 * \n + 2 * \x, 0) to 
  [out=-90, in=90] (-3 + 4 * \n + 2 * \y, -\offset);  
  \draw[-{Latex}, \col, very thick] (-1 + 2 * \n + 2 * \y, -\offset) to 
  [out=90, in=-90] (-3 + 4 * \n + 2 * \x, 0);
  
  
  
  \def \x {3};
  \def \y {1};
  \def \col {DarkOrchid};
  
  \draw[-{Latex}, \col, very thick] (-1 + 2 * \n + 2 * \x, 0) to 
  [out=-90, in=90] (-3 + 4 * \n + 2 * \y, -\offset);  
  \draw[-{Latex}, \col, very thick] (-1 + 2 * \n + 2 * \y, -\offset) to 
  [out=90, in=-90] (-3 + 4 * \n + 2 * \x, 0);

  

  \def \x {3};
  \def \y {4};
  \def \col {Cyan};
  
  \draw[-{Latex}, \col, very thick] (-1 + 2 * \n + 2 * \x, 0) to 
  [out=-90, in=90] (-3 + 4 * \n + 2 * \y, -\offset);  
  \draw[-{Latex}, \col, very thick] (-1 + 2 * \n + 2 * \y, -\offset) to 
  [out=90, in=-90] (-3 + 4 * \n + 2 * \x, 0);
  
  

  
  
\end{tikzpicture}
  \caption{Snippet of the construction
  of the gadget (bottom) for the graph $G$
  of non-edges (top) with $n = 4$ and $k = 2$.
  Vertices $x_0,\dots,x_{16}$ and
  $y_0,\dots,y_{16}$ are junction vertices
  of the necklaces corresponding to the left
  hand side and the right hand side of $G$,
  respectively.
  Colorful arrows in the gadget
  represent paths of crossing arcs.
  Solid delimiters at $x_0$, $x_{12}$ and $y_0$, $y_{12}$
  indicate starts of new strings.
  }
  \label{fig:gadget}
\end{figure}

\paragraph*{Coordination gadgets}

Now we construct the coordination gadgets
that disallow choosing non-adjacent vertices.
Consider a pair of non-adjacent 
vertices $v^{i}_{\alpha}$ and $v^{j}_{\beta}$ in $G$ with $i < j$.
Let the junction vertices in 
strings $S^{i}_{j}$ and $S^{j}_{i}$ be
$x^{i,j}_0, \dots, x^{i,j}_{3n-1}$ and 
$y^{j,i}_0, \dots, y^{j,i}_{3n-1}$,
respectively.
Introduce vertices 
$s^{i,j}_{\alpha,\beta}$ and $t^{j,i}_{\alpha,\beta}$
and add four undeletable arcs forming
directed paths
\[ x^{i,j}_{n + \alpha} \to s^{i,j}_{\alpha,\beta} \to y^{j,i}_{2n + \beta - 1}
\quad \textnormal{ and } \quad
y^{j,i}_{n + \beta} \to t^{j,i}_{\alpha,\beta} \to x^{i,j}_{2n + \alpha - 1}. \]
Call these four arcs \emph{crossing}.
Add cut request $\{s^{i,j}_{\alpha,\beta}, t^{j,i}_{\alpha,\beta}\}$ to $\cP$.
Observe that this request
requires the solution
to pick a diamond between
between $x^{i,j}_{n + \alpha}$ and $x^{i,j}_{2n + \alpha - 1}$
or a diamond between
$y^{j,i}_{n + \beta}$ and
$y^{j,i}_{2n + \beta - 1}$.

The construction is complete (see Figure~\ref{fig:gadget}).
We proceed with the correctness proof.

\paragraph*{Directed Symmetric Multicut to Clique}

Suppose $X$ is a solution to $(D, \cP, k')$.
Let $X' \subseteq V(G)$ contain the vertices
of $G$ picked by $X$, i.e.
$X'$ contains $v^i_\alpha$
whenever $\lozenge^{i,1}_{\alpha}$ is in $X$.
Note that $|X'| = k$
by the observation 
that $\alpha$ is unique to the set $V^i$.
We claim that $X'$ forms a clique in $G$.
Suppose for the sake of contradiction that
two non-adjacent vertices
$v^{i}_{\alpha}, v^{j}_{\beta}$ are in $X'$.
Let the junction vertices in $S^{i}_{j}$ and $S^{j}_{i}$ 
be $x^{i,j}_0, \dots, x^{i,j}_{3n-1}$
and
$y^{j,i}_0, \dots, y^{j,i}_{3n-1}$, respectively.
By construction, $X$ picks diamonds 
$\lozenge^{i,j}_{\alpha}$,
$\lozenge^{i,j}_{\alpha+n}$, 
$\lozenge^{i,j}_{\alpha+2n}$
in $S^{i}_j$ and
$\lozenge^{j,i}_{\beta}$,
$\lozenge^{j,i}_{\beta+n}$,
$\lozenge^{j,i}_{\beta+2n}$ in $S^{j}_i$.
Then $D - X$ contains a closed walk
\[ x^{i,j}_{n + \alpha} \to s^{i,j}_{\alpha,\beta} \to y^{j,i}_{2n + \beta - 1} 
\to \dots \to y^{j,i}_{n + \beta} \to t^{j,i}_{\alpha,\beta} \to x^{i,j}_{2n + \alpha - 1} 
\to \dots \to x^{i,j}_{n + \alpha}, \]
contradicting that $\{s^{i,j}_{\alpha,\beta}, t^{j,i}_{\alpha,\beta}\}$ is satisfied.

\paragraph*{Clique to Directed Symmetric Multicut}
Suppose $Z \subseteq V(G)$ induces a 
complete subgraph in $G$,  and 
$|Z \cap V^i| = 1$ for all $1 \leq i \leq k$.
Define a set of arcs $Z'$ in $D$ by picking 
diamonds $\lozenge^{i,j}_{\alpha}$,
$\lozenge^{i,j}_{\alpha+n}$,
$\lozenge^{i,j}_{\alpha+2n}$
for all $v^i_\alpha \in Z$
and $1 \leq j \leq k$.
Clearly, $|Z'| = 3k|Z| = k'$.
Further, we show that $Z'$ is a solution to 
$(D, \cP, k')$.
To this end, we prove an auxiliary claim, 
and show that it
implies that every cut request in $\cP$
is satisfied in $D - Z'$

Observe that $Z'$ partitions each necklace
into $3k$ strongly connected components
which we call \emph{runs}.
We claim that 
no two neighbouring runs are 
strongly connected in $D - Z'$.
It suffices to show that
for every pair of consecutive runs,
there is no closed walk containing
vertices from both of them.
Note that a counterexample 
would need to contain a crossing arc
because chosen diamonds
separate runs within the necklace.
By construction,
junction vertices $x_0, \dots, x_{3n-1}$
in any string $S^i_j$
are of one of the following types
with respect to the crossing arcs:
\begin{itemize}
  \item $x_0, \dots, x_{n}$ admit 
  admit neither incoming nor outgoing crossing arcs,
  \item $x_{n+1}, \dots, x_{2n-1}$ admit only 
  outgoing crossing arcs, 
  \item $x_{2n}$ admits both incoming and outgoing crossing arcs,
  and
  \item $x_{2n+1}, \dots, x_{3n-1}$ admit
  only incoming crossing arcs.
\end{itemize}
Moreover, since the construction is periodic, 
$n+1$ vertices following 
$x_{3n-1}$ also admit no incoming or outgoing crossing arcs.
Chosen diamonds are evenly spaced,
so a run contains $n$ junction vertices.
Note that if a run admits both incoming and outgoing crossing arcs,
then it contains $x_{2n}$ and one of 
its neighbouring runs admits only outgoing
and another -- only incoming crossing arcs.
Hence, no pair among these three runs is strongly connected.

With this observation, we first verify that
every cut request within a necklace is satisfied by $Z'$.
Indeed, endpoints of cut requests in $D - Z'$ are in neighbouring runs, so they are not strongly connected.
Now consider cut requests 
$\{s^{i,j}_{\alpha\beta}, t^{j,i}_{\alpha\beta}\}$
for every pair of non-adjacent 
vertices $v^{i}_{\alpha}$ and $v^{j}_{\beta}$ in $G$.
Let the junction vertices on strings $S^i_j$ and $S^j_i$
be
$x^{i,j}_0, \dots, x^{i,j}_{3n-1}$ and $y^{j,i}_0, \dots, y^{j,i}_{3n-1}$, respectively.
Four endpoints of the crossing arcs 
incident to $s^{i,j}_{\alpha,\beta}$ and $t^{i,j}_{\alpha,\beta}$
are $x^{i,j}_{n + \alpha}, x^{i,j}_{2n + \alpha - 1}$ and $y^{j,i}_{n + \beta}, y^{j,i}_{2n + \beta - 1}$.
We claim that either
$x^{i,j}_{n + \alpha}$ and $x^{i,j}_{2n + \alpha - 1}$ are not strongly connected
or $y^{j,i}_{n + \beta}$ and $y^{j,i}_{2n + \beta - 1}$ are not strongly connected.
Recall that $Z$ induces a clique in $G$,
and $v^i_\alpha, v^j_\beta$ are non-adjacent in $G$,
hence either $v^i_\alpha \notin Z$ or $v^j_\beta \notin Z$.
Without loss of generality, assume that $v^i_\alpha \notin Z$,
and observe that $Z$ contains $v^i_{\alpha'}$ for some $\alpha' \neq \alpha$.
Then $Z'$ chooses a diamond that lies between $x^{i,j}_{n + \alpha}$ and $x^{i,j}_{2n + \alpha - 1}$,
hence these two vertices 
end up in neighbouring runs and are not strongly connected.
The case with $v^j_{\beta} \notin Z$ is symmetric.

\section{Concluding remarks}
\label{sec:conclusion}

We classified parameterized complexity of $\wmincsp{\Gamma}$ for all subsets $\Gamma$ of Point Algebra, i.e. $\{<,\leq,=,\neq\}$ on domain $\QQ$.
In particular, we prove that $\mincsp{\leq,\neq}$ is W[1]-hard,
settling the complexity of \textsc{Directed Symmetric Multicut}. 
DSMC was a roadblock to classifying interval and temporal constraint languages,
i.e. first-order generalizations of Point Algebra
and Allen's Interval Algebra, respectively.
Tackling these two classifications is a natural continuation.
There are several difficulties: for interval constraint languages, even a polynomial-time CSP dichotomy is still unavailable (see~\cite{bodirsky2022complexity}), 
and this is a prerequisite for a \textsc{MinCSP} classification.
A more approachable fragment is the set of all binary interval relations for which the dichotomy is known~\cite{Krokhin:etal:jacm2003}.

For temporal constraints, a CSP dichotomy is known~\cite{bodirsky2010complexity}.
However, if we compare with Boolean and equality languages,
the only ones for which parameterized \textsc{MinCSP} dichotomies have been obtained,
the temporal CSP classification is much less
manageable -- the Boolean dichotomy is Schaefer's celebrated
result from 1978~\cite{schaefer1978complexity}, while the equality CSP dichotomy~\cite{bodirsky2008complexity} 
has only one nontrivial tractable class.
In contrast, there are ten nontrivial tractable fragments of 
temporal relations defined by their
algebraic invariants~\cite[Theorem~50]{bodirsky2010complexity}.

On the algorithmic side, we want to highlight $\mincsp{<,=}$. 
By reduction to \textsc{Subset-DFAS} and using the best known algorithm~\cite{chitnis2015directed},
one can solve this problem in $O^*(2^{O(k^3)})$ time.
Can this running time be improved?
We remark that such an improvement would speed up the $2$-approximation 
algorithm for fixed-parameter tractable
\textsc{MinCSP}s over basic interval relations~\cite{dabrowski2023parameterized}.

\section*{Acknowledgements}

This work was carried out during the Copenhagen Summer of Counting \& Algebraic Complexity, funded by research grants from VILLUM FONDEN (Young Investigator Grant 53093) and the European Union (ERC, CountHom, 101077083). Views and opinions expressed are those of the authors only and do not necessarily reflect those of the European Union or the European Research Council Executive Agency. Neither the European Union nor the granting authority can be held responsible for them.

George was supported by the Wallenberg AI, Autonomous Systems and Software Program (WASP) funded by the Knut and Alice Wallenberg Foundation.
During this research Marcin was part of BARC, supported by the VILLUM Foundation grant 16582.

\bibliographystyle{plain}
\bibliography{references}

\end{document}